\documentclass[letterpaper, 10 pt, conference]{ieeeconf}  
\usepackage{xcolor}

\IEEEoverridecommandlockouts                              
\usepackage[colorinlistoftodos,prependcaption, linecolor=red,backgroundcolor=white,bordercolor=white, textcolor=red, textsize=small]{todonotes}

\usepackage{amssymb}
\let\cssconfproof\proof
\let\cssconfendproof\endproof
\let\proof\relax
\let\endproof\relax
\usepackage{amsthm}
\let\proof\cssconfproof
\let\endproof\cssconfendproof
\usepackage{amsmath}
\usepackage[separate-uncertainty = true,multi-part-units=single]{siunitx}
\usepackage{xfrac}
\usepackage{cancel}
\usepackage{bm}
\usepackage{mathtools}

\usepackage[font=footnotesize, caption=false]{subfig}

\usepackage{xcolor}
\usepackage{hyperref}
\usepackage[capitalise]{cleveref} 
\usepackage{autonum}

\usepackage[noadjust]{cite}

\usepackage{setspace}
\usepackage{pgf}

\usepackage{xparse}

\NewDocumentCommand\bbm{}{ \begin{bmatrix} }
\NewDocumentCommand\ebm{}{ \end{bmatrix} }
\NewDocumentCommand\bpm{}{ \begin{pmatrix} }
\NewDocumentCommand\epm{}{ \end{pmatrix} }

\NewDocumentCommand\Real{}{ \mathbb{R} }

\newif\ifcomments
\commentsfalse
\newcommand{\ah}[1]{%
        \ifcomments
            {\color{red}{(#1)}}%
        \else
            \relax
        \fi
}

\newcommand{\new}[1]{{\color{black}{#1}}}

\NewDocumentCommand\Mean{m}{\pmb{\mu}^{m}}

\NewDocumentCommand\Cov{m}{\pmb{\Sigma}^{m}}

\newcommand{\III}[1]{{\left\vert\kern-0.25ex\left\vert\kern-0.25ex\left\vert #1 \right\vert\kern-0.25ex\right\vert\kern-0.25ex\right\vert}}
\newcommand{\IIIf}[1]{{\vert\kern-0.25ex\vert\kern-0.25ex\vert #1 \vert\kern-0.25ex\vert\kern-0.25ex\vert}}
\renewcommand{\phi}{\varphi}

\renewcommand{\epsilon}{\varepsilon}

\newcommand{\tran}{\mathrm{T}}

\DeclareMathOperator{\diag}{diag}

\renewcommand{\vec}[1]{\bm{#1}} 
\newcommand{\mat}[1]{\bm{#1}} 

\newcommand{\dotGen}[2]{#2^{(#1)}}
\newcommand{\dotThree}[1]{\dotGen{3}{#1}}

\newcommand{\regState}{\vec{x}}
\newcommand{\regInput}{\vec{u}}
\newcommand{\regDisInput}{\regInput_k}
\newcommand{\regExtInput}{\bar{\regInput}} 

\newcommand{\mpcQ}{\mat{Q}}
\newcommand{\mpcR}{\mat{R}}
\newcommand{\mpcHorizon}{T}

\newcommand{\flatOutput}{\vec{y}} 
\newcommand{\flatOutputComp}[1]{\prescript{}{#1}{y}}
\newcommand{\flatState}{\vec{z}}
\newcommand{\flatInput}{\vec{v}}

\newcommand{\flatInputTrans}{\vec{\Gamma}}

\newcommand{\flatStateTrans}{\vec{\Phi}}
\newcommand{\flatOutputDef}{\vec{\Lambda}}

\newcommand{\gpFlatInput}[1]{\prescript{}{#1}{v}}
\newcommand{\gpFunToLearn}[1]{\prescript{}{#1}{\Psi}}
\newcommand{\gpFunToLearnVec}{\vec{\Psi}}
\newcommand{\gpInputData}{\vec{a}}

\newcommand{\kernel}{\bar{k}}
\newcommand{\kernAlpha}{k_{\alpha}}
\newcommand{\kernBeta}[1]{k_{\beta, #1}}

\newcommand{\gpMean}[1]{\prescript{#1}{}{\hspace{-1pt}\mu}}
\newcommand{\gpCov}[1]{\prescript{#1}{}{\hspace{-1pt}\sigma^{2}}}
\newcommand{\gpStddev}[1]{\prescript{#1}{}{\sigma}}

\newcommand{\gamOne}[1]{\prescript{#1}{}{\gamma_{1}}}
\newcommand{\gamTwo}[1]{\prescript{#1}{}{\vec{\gamma}_{2}}}
\newcommand{\gamThree}[1]{\prescript{#1}{}{\gamma_{3}}}
\newcommand{\gamFour}[1]{\prescript{#1}{}{\vec{\gamma}_{4}}}
\newcommand{\gamFive}[1]{\prescript{#1}{}{\mat{\gamma}_{5}}}

\newcommand{\gpTargetVec}{\vec{\hat{\Psi}}}
\newcommand{\gpNData}{N} 

\newcommand{\flatInpOptComp}[1]{\prescript{}{#1}{v}_k^*}
\newcommand{\flatStateOpt}{\flatState_k^*}
\newcommand{\flatInputOpt}{\flatInput_k^*}
\newcommand{\regDisExtInp}{\regExtInput_k} 

\newcommand{\gamOneOpt}[1]{{\gamOne{#1}^{\hspace{-0.25em}*}}}
\newcommand{\gamTwoOpt}[1]{{\gamTwo{#1}^{\hspace{-0.25em}*}}}
\newcommand{\gamThreeOpt}[1]{{\gamThree{#1}^{\hspace{-0.25em}*}}}
\newcommand{\gamFourOpt}[1]{{\gamFour{#1}^{\hspace{-0.25em}*}}}
\newcommand{\gamFiveOpt}[1]{{\gamFive{#1}^{\hspace{-0.25em}*}}}

\newcommand{\socpOptVar}{\tilde{\vec{u}}} 
\newcommand{\socA}[1]{\bar{\mat{A}}_{#1}}
\newcommand{\socB}[1]{\bar{\vec{b}}_{#1}}
\newcommand{\socC}[1]{\bar{\vec{c}}_{#1}}
\newcommand{\socD}[1]{\bar{d}_{#1}}

\newcommand{\socBComp}[2]{\bar{\vec{b}}_{#1, #2}}

\newcommand{\choleskyGamFive}[1]{\prescript{#1}{}{\mat{L}^*}}

\newcommand{\flatDisState}[1]{\flatState_{#1}}
\newcommand{\flatDisInput}[1]{\flatInput_{#1}}
\newcommand{\flatDisStateRef}[1]{{\flatState_{#1}^{\footnotesize \text{ref}}}}
\newcommand{\flatDisInputRef}[1]{{\flatInput_{#1}^{\footnotesize \text{ref}}}}
\newcommand{\trackingError}[1]{\vec{e}_{#1}}
\newcommand{\ljapunovFun}[1]{V(#1)}
\newcommand{\stabPmat}{\mat{P}}

\newcommand{\lqrGain}{\mat{K}}

\newcommand{\dynAdis}{\mat{A}_d} 
\newcommand{\dynBdis}{\mat{B}_d}

\newcommand{\flatNomInput}{\flatInput_k^{nom}}
\newcommand{\flatNomInpComp}[1]{{\prescript{}{#1}{v}_k^{nom}}}

\newcommand{\stabWone}{\vec{w}_1}

\newcommand{\stabWthree}{w_3}

\newcommand{\stabWtwoMat}{\mat{W}_2}
\newcommand{\stabWtwoMatComp}[1]{\prescript{}{#1}{W_2}}
\newcommand{\stabWfour}{\vec{w}_4}
\newcommand{\stabWfourComp}[1]{\prescript{}{#1}{w_4}}
\newcommand{\lipschitzConst}[1]{L_{#1}}

\newcommand{\stabProb}[1]{\delta_{#1}}
\newcommand{\stabBetaSqrt}[1]{\beta_{#1}^{1/2}}
\newcommand{\stabFunComp}{\prescript{}{i}{g}}

\newcommand{\dynExtState}{\vec{\eta}}

\newcommand{\dynExtSelectionMat}{\mat{C}_{ext}}
\newcommand{\dynExtAmatDis}{\mat{A}_{d, ext}}
\newcommand{\dynExtBmatDis}{\mat{B}_{d, ext}}

\newcommand{\meanState}[1]{\vec{\mu}_{\flatState, #1}}
\newcommand{\covState}[1]{\vec{\Sigma}_{\flatState, #1}}

\newcommand{\dynBdisComp}[1]{\vec{b}_{d, #1}}

\newcommand{\stateSelVec}{\vec{h}_j}
\newcommand{\stateSelVecPart}[1]{\prescript{}{#1}{\vec{h}_j}}
\newcommand{\stateLimit}{b_j}
\newcommand{\stateTightenedSet}{\mathcal{Z}^{\mathcal{R}}}
\newcommand{\stateW}[1]{w_{s, #1}}

\definecolor{tum blue}{HTML}{0065BD}

\hypersetup{
  colorlinks      = false, 
  allbordercolors = tum blue!20, 
}

\crefformat{equation}{{#2(#1)#3}}
\Crefformat{equation}{{#2(#1)#3}}
\crefname{chapter}{Chapter}{Chapters}
\crefname{section}{Section}{Sections}
\crefname{subsection}{Section}{Sections}
\crefname{appendix}{\S}{\S}

\newtheorem{dfn}{Definition}[section]
\newtheorem{thm}{Theorem}[section]
\newtheorem{lem}{Lemma}[section]
\newtheorem{rem}{Remark}[section]
\newtheorem{asm}{Assumption}[section]

\makeatletter
\def\thm@space@setup{%
  \thm@preskip=1pt plus 1pt minus 1pt
  \thm@postskip=1pt plus 1pt minus 1pt}
\makeatother

\newcommand{\smallvdots}{\raisebox{-2pt}{\scalebox{0.8}{$\vdots$}}}

\title{\LARGE \bf
Exploiting Differential Flatness for Efficient Learning-based Model Predictive Control \new{of Constrained Multi-Input Control Affine Systems} 
}

\author{Tobias A. Farger$^{1,*}$, Adam W. Hall$^{2,*}$, and Angela P. Schoellig$^1$
\thanks{$^*$ These authors contributed equally to the manuscript.}
\thanks{$^1$ Learning Systems and Robotics Lab at the Technical University of Munich, Munich, Germany. }%
\thanks{$^2$ Learning Systems and Robotics Lab at the University of Toronto Institute for Aerospace Studies (UTIAS) and the Vector Institute for Artificial Intelligence, Toronto, Canada}%
\thanks{Email: {\tt\small \{tobias.farger, angela.schoellig\}@tum.de} {\tt\small adam.hall@robotics.utias.utoronto.ca}}%
\thanks{This work was supported by the Robotics Institute Germany under BMBF grant 16ME0997K.}
}

\newcommand{\secVspace}{\vspace{-1mm}}
\begin{document}
\maketitle
\thispagestyle{empty}
\pagestyle{empty}
\bstctlcite{BSTcontrol}
\begin{abstract}
Learning-based control techniques use data from past trajectories to  control systems with uncertain dynamics. 
However, learning-based controllers are often computationally inefficient, limiting their practicality. 
To address this limitation, we propose a learning-based controller that exploits differential flatness, a property of many robotic systems.
Recent research on using flatness for learning-based control either is limited in that it (i) ignores input constraints, (ii) applies only to single-input systems, or (iii) is tailored to specific platforms.
In contrast, our approach uses a system extension and block-diagonal cost formulation to control general multi-input, nonlinear, affine systems.
Furthermore, it satisfies input and half-space flat state constraints and guarantees probabilistic Lyapunov decrease using only two sequential convex optimizations. 
\new{We show that our approach performs similarly to, but is multiple times more efficient than, a Gaussian process model predictive controller in simulation, and achieves competitive tracking in real hardware experiments.}
\end{abstract}

\secVspace
\section{Introduction}
\secVspace
Learning-based control methods use machine learning models to estimate uncertain dynamics using data collected from previous trajectories. 
Recently, such methods have proven to be a less conservative, but more computationally expensive, option for controlling uncertain systems as compared to classical robust control \cite{brunkeSafeLearningRobotics2022}, \new{particularly when accurate first-principles models are difficult to obtain or system parameters vary across operating conditions}.
A common approach uses Gaussian processes (GPs) to model the dynamics within a model predictive control (MPC) framework \cite{hewingCautiousModelPredictive2020}, which necessitates solving a nonlinear program at every time step---making real-time computation challenging.  
One can, however, use structural assumptions about the system dynamics, like differential flatness \cite{fliessFlatnessDefectNonlinear1995}, to improve the efficiency of learning-based controllers. 

Differential flatness, a property of many robotic systems (e.g., quadrotors \cite{greeffFlatnessBasedModelPredictive2018} and flexible-joint manipulators \cite{isidoriNonlinearControlSystems1995}), allows for linear control techniques to be applied to nonlinear systems.
This is because differentially flat systems can be exactly transformed into linear systems via nonlinear input transformations. 
Thus, linear control methods, like linear MPC, can be used to control flat nonlinear systems. 
Such flat model predictive control (FMPC) methods perform similarly to nonlinear model predictive control (NMPC) but require less compute~\cite{greeffFlatnessBasedModelPredictive2018,sunComparativeStudyNonlinear2022}. FMPC, however, cannot enforce input constraints while maintaining convexity and is sensitive to model mismatch in the nonlinear transformation.

\begin{figure}[t]
    \centering
    \includegraphics[width=\linewidth]{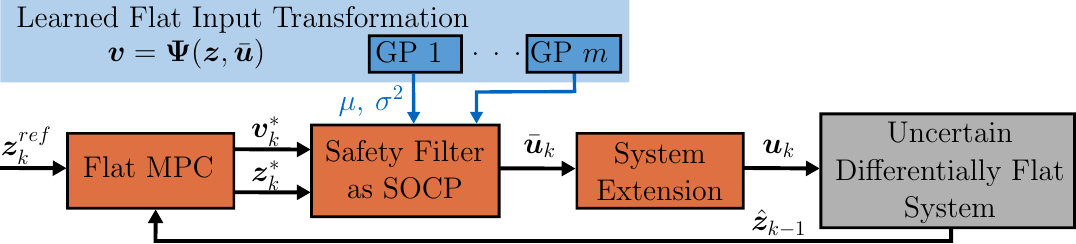}
    \caption{ 
    Our approach consists of a flat model predictive controller (MPC), a safety filter and a system extension.
    The flat MPC tracks the reference $\flatDisStateRef{k}$ and computes the optimal flat input $\flatInputOpt$ and state $\flatStateOpt$. 
    As the nonlinear flat input transformation $\flatInput = \vec{\Psi}(\flatState, \regExtInput)$ is uncertain, \new{it} is modeled with an independent Gaussian processes (GPs) for each component. 
    The second-order cone program (SOCP) safety filter then uses the GPs' mean and covariance predictions $\gpMean{}$ and $\gpCov{}$, along with $\flatStateOpt$ and $\flatInputOpt$, to find the extended input $\regDisExtInp$ that guarantees probabilistic Lyapunov decrease and flat state constraint satisfaction while respecting constraints on the input $\regDisInput$. }
    \label{fig:block_diag}
\end{figure}

To eliminate any model mismatch and incorporate input constraints, recent work has explored learning-based controllers that exploit differential flatness. In \cite{greeffExploitingDifferentialFlatness2021}, a robust linear quadratic regulator is combined with a learned nonlinear input mapping and its inverse for single-input systems. While this approach provides an upper bound on tracking error, it lacks the ability to explicitly handle input or flat state constraints. In contrast, \cite{akbariComputationallyEfficientLearningBased2024} and \cite{akbariTinyLearningBasedMPC2024} introduce a method tailored to multi-rotor systems, capable of handling multiple inputs and respecting input constraints. Although these techniques are efficient, they are limited to multi-rotor systems.


Other recent approaches combine FMPC with a second order cone program (SOCP) safety filter that uses GPs to learn the nonlinear input transformation~\cite{greeffLearningStabilityFilter2021, hallDifferentiallyFlatLearningBased2023}. 
This safety filter complements the FMPC by guaranteeing probabilistic Lyapunov decrease, probabilistic half-space flat state constraint satisfaction, and bounded input constraint satisfaction.
However, these approaches are limited to single-input systems, and their input transformations do not require input derivatives (i.e., for systems with full relative degree \cite{isidoriNonlinearControlSystems1995}).

The contributions of this work are:
\begin{itemize}
    \item A reformulated SOCP safety filter that supports general multi-input, control-affine systems with flat input transformations requiring input derivatives;
    \item A novel Lipschitz-based bound and block-diagonal cost structure enabling probabilistic Lyapunov decrease;
    \item Simulation and hardware validation on quadrotor trajectory tracking, showing competitive performance with GPMPC and NMPC, but with significantly reduced computation time.
\end{itemize}
In contrast to prior work limited to single-input settings or quadrotor-specific formulations, we demonstrate a learning-based, flatness-exploiting MPC that scales to general multi-input systems while satisfying flat-state and input constraints with real-time computation.
As in \cite{hallDifferentiallyFlatLearningBased2023}, our formulation still ensures probabilistic Lyapunov decrease and satisfaction of half-space flat state and input constraints for the next time step.
Although the formulation does not guarantee recursive feasibility, we show through simulation and hardware experiments that the controller remains feasible, stable, and constraint-satisfying in practice.
\new{Providing such guarantees for computationally efficient learning-based controllers is critical for deploying robotic systems that must operate in real time under uncertain dynamics and safety constraints.}

%
\section{Problem Statement}
\secVspace
Consider a nonlinear multi-input control-affine system,
\begin{equation}
    \dot{\regState}(t) = \vec{f}(\regState(t)) + \vec{g}(\regState(t)) \regInput,
    \label{eq:dynamics}
\end{equation}
with initial condition $\regState(0) = \regState_0$. Here $\regState(t) \in \Real^{n_x}$ and $\regInput(t) \in \Real^m$ denote the system state and input, respectively.
The unknown functions $\vec{f}: \Real^{n_x} \rightarrow \Real^{n_x}$ and $\vec{g}: \Real^{n_x} \rightarrow \Real^{n_x \times m}$ are assumed to be sufficiently smooth and \eqref{eq:dynamics} is known to be \textit{differentially flat} with a known \textit{flat output}.

\begin{dfn}[Differential Flatness \cite{levineAnalysisControlNonlinear2009}]
    A system \cref{eq:dynamics} is differentially flat with respect to a flat output $\flatOutput = \flatOutputDef(\regState, \regInput, \dot{\regInput},..., \dotGen{\ell}{\regInput})$ if all system states and inputs can be written as $\regState = \flatStateTrans(\flatOutput, \dot{\flatOutput}, ..., \dotGen{\rho_x}{\flatOutput})$ and $\regInput = \flatInputTrans(\flatOutput, \dot{\flatOutput}, ..., \dotGen{\rho_x + 1}{\flatOutput})$.
\end{dfn}
The input $\regInput$ is subject to constraints $\regInput_{min} \leq \regInput \leq \regInput_{max}$ (all vector inequalities are evaluated component-wise). 
In addition, we consider convex half-space constraints on the flat state $\flatState \in \mathcal{Z} = \{ \flatState | \stateSelVec^\tran \flatState \leq \stateLimit, \forall j=1\ldots n_c \} \subset \Real^{n_z}$, where $\flatState = [\flatOutputComp{1}, ..., \dotGen{\rho_1 -1}{\flatOutputComp{1}}, ...,  \flatOutputComp{m}, ..., \dotGen{\rho_m -1}{\flatOutputComp{m}}] ^\tran$, and $n_z = \sum_i (\rho_i-1)$.
We also have that $\rho_x = \max_i(\rho_i-1)$. 
The subscript $i$ at the bottom left indicates component $i$ of a vector and $\dotGen{p}{(\cdot)}$ denotes the derivative of order $p$.
\begin{rem}
    Although we only consider constraints on the flat states $\flatState$, these are often useful quantities like position, velocity, acceleration, and attitude.
\end{rem}
The objective is to design a computationally efficient controller for \cref{eq:dynamics} that ensures accurate tracking of a flat state reference $\flatDisStateRef{}$ and satisfies input and flat state constraints. 

%
\secVspace
\section{Background}
\secVspace
\subsection{Discretized Exact Linear Dynamics}
Differential flatness implies exactly linear dynamics in a transformed flat-state space, allowing linear controllers to be designed in the flat space and mapped back to the original nonlinear system.
\begin{lem}[Linearized Flat Dynamics \cite{levineAnalysisControlNonlinear2009,isidoriNonlinearControlSystems1995}]
    Every flat system can be represented as a linear system and a nonlinear transformation
\begin{align}
    \dot{\flatState} &= \mat{A} \flatState + \mat{B} \flatInput \label{eq:linSys} \\
    \flatInput &= \vec{\Psi}(\flatState, \regInput, \dot{\regInput}, ..., \dotGen{\ell}{\regInput}), \label{eqn:flatInputPsi}
\end{align}
where $\flatInput = [\flatOutputComp{1}^{(\rho_1)},\ldots,\flatOutputComp{m}^{(\rho_m)}]^\tran \in \Real^m$ 
is the flat input, and $\mat{A}$ and $\mat{B}$ are such that \cref{eq:linSys} is a system of chained integrators.
\end{lem}

In this work, we discretize the flat system \cref{eq:linSys} such that
\begin{equation}
    \flatDisState{k+1} = \dynAdis \flatDisState{k} + \dynBdis \flatDisInput{k}, 
    \label{eq:flatSysDiscretized}
\end{equation}
where $k$ is the discrete time index, $\flatDisState{k}$ is the discretized flat state, $\flatDisInput{k}$ is the discretized flat input, and $\dynAdis$ and $\dynBdis$ are exact discretizations of $\mat{A}$ and $\mat{B}$, respectively. 
Since \cref{eq:linSys} decomposes into $m$ independent integrator chains, $\dynAdis$ is block diagonal and $\dynBdis = \diag (\dynBdisComp{1}, ..., \dynBdisComp{m})$ with $\dynBdisComp{i} \in \Real^{\rho_i}$.
Although this discretization introduces a sampling error,
we assume that, if the sampling time is sufficiently small, this error can be easily accounted for by constraint tightening.
Such discretizations are standard in the flat control literature~\cite{akbariComputationallyEfficientLearningBased2024, akbariTinyLearningBasedMPC2024, greeffExploitingDifferentialFlatness2021, greeffFlatnessBasedModelPredictive2018, hallDifferentiallyFlatLearningBased2023, sunComparativeStudyNonlinear2022}.

\subsection{Gaussian Processes}
\secVspace
We use GPs to model distributions over unknown but observable functions $\Psi(\gpInputData): \Real^{\text{dim}(\gpInputData)} \rightarrow \Real$. 
Prior knowledge about $\Psi(\gpInputData)$ is encoded via a prior mean and parameterized kernel function. 
Then, given a noisy data set $\hat{\Psi}(\gpInputData) = \Psi(\gpInputData) + \epsilon$ with $\epsilon \sim \mathcal{N}(0,\sigma_n^2)$, the GP's hyperparameters are regressed to maximize the log-likelihood of the marginal distribution over the sampled data~\cite{rasmussenGaussianProcessesMachine2006}.

For a query point $\gpInputData^*$ the posterior prediction conditioned on the data $\mathcal{D} = \{ (\gpInputData_i, \hat{\Psi}(\gpInputData_i)) \}_{i = 1}^{\gpNData}$ is a normal distribution $\Psi(\gpInputData^*) | \mathcal{D} = \mathcal{N}( \gpMean{}(\gpInputData^*), \gpCov{} (\gpInputData^*))$, where the posterior mean $\gpMean{}$ and covariance $\gpCov{}$ are computed as $\gpMean{}(\gpInputData^*) = \vec{k}(\gpInputData^*) \mat{K}^{-1} \gpTargetVec$ and $\gpCov{}(\gpInputData^*) = \kernel(\gpInputData^*, \gpInputData^*) - \vec{k}(\gpInputData^*) \mat{K}^{-1} \vec{k}(\gpInputData^*) ^\tran$, where $\kernel$ is the kernel of the GP, $\vec{k}(\gpInputData^*) = [\kernel(\gpInputData^*, \gpInputData_1), ..., \kernel(\gpInputData^*, \gpInputData_\gpNData)]$, $\mat{K}_{i, j} = \kernel(\gpInputData_i, \gpInputData_j) + \sigma_n^2 \delta_{i,j}$, where $\delta_{i,j}$ is the Kronecker delta, and $\gpTargetVec = [\hat{\Psi}(\gpInputData_1), ..., \hat{\Psi}(\gpInputData_\gpNData)]$. 

\subsection{Discrete-Time Lyapunov Function}
\secVspace
Consider the discrete-time dynamics \cref{eq:flatSysDiscretized} and a smooth reference  $\flatDisStateRef{k}, \flatDisInputRef{k}$. The tracking error with respect to this reference is defined as $\trackingError{k} = \flatDisState{k} - \flatDisStateRef{k}$. 
We have the following standard stability result~\cite[Thm. 5.5]{ogataDiscretetimeControlSystems1998}.
\begin{lem}[Discrete-Time Lyapunov Decrease]
Consider the feedback $\flatDisInput{k} = - \lqrGain \trackingError{k} + \flatDisInputRef{k} $ and the resultant feedback error dynamics $\trackingError{k+1} = ( \dynAdis - \dynBdis \lqrGain) \trackingError{k}$.
If there exists a function $V : \Real^{n_z} \rightarrow \Real_{\geq 0}$ that satisfies 
  \begin{equation}
    \begin{aligned}
    &\ljapunovFun{\vec{0}} = 0, \; \ljapunovFun{\trackingError{k}} > 0, \; 
    \ljapunovFun{\trackingError{k+1}} < \ljapunovFun{\trackingError{k}},
  \end{aligned}
  \label{eq:ljapunovDecreasing}
  \end{equation}
    $\forall \trackingError{k} \in \Real^{n_z} \setminus \{\mathbf{0}\}$, then the closed-loop error dynamics are asymptotically stable with respect to the origin.
\end{lem}

%
\secVspace
\section{Methodology}
\secVspace
Our proposed architecture (\cref{fig:block_diag}) builds upon \cite{hallDifferentiallyFlatLearningBased2023}. The control approach consists of an FMPC, a safety filter, and a system extension. 
The FMPC (\cref{sec:FMPC}) computes the optimal flat state and input to best track $\flatDisStateRef{}$ while respecting flat-state constraints. 
As the nonlinear input transformation is unknown, we learn it offline from data using an individual GP for each component of $\flatInput$ (\cref{sec:GPaffine}). 
The safety filter (\cref{sec:safetyFilter}) then uses the GPs to obtain the extended system input $\regDisExtInp$ while satisfying probabilistic Lyapunov decrease, probabilistic half-space flat state constraints, and constraints on the extended input $\regDisExtInp$ and system input $\regDisInput$. 
Finally, we obtain $\regDisInput$ from $\regDisExtInp$ in the system extension (\cref{sec:systemExtension}). 
For control-affine systems, the safety filter can be written as an SOCP. As the FMPC is a quadratic program, both optimization problems are convex and can be solved efficiently.
While recursive feasibility is not theoretically guaranteed, we show in our experiments that, pragmatically, feasibility is maintained.

\secVspace
\subsection{Flat Model Predictive Control}
\secVspace
\label{sec:FMPC}

FMPC solves a discrete-time, finite-horizon convex optimal control problem (OCP). The formulation we introduce here builds on \cite{greeffFlatnessBasedModelPredictive2018} and is equivalent to the one presented in \cite{hallDifferentiallyFlatLearningBased2023}. 
We consider the cost $(\flatDisState{k}-\flatDisStateRef{k})^\tran \mpcQ (\flatDisState{k}-\flatDisStateRef{k}) + \flatDisInput{k}^\tran \mpcR \flatDisInput{k}$, where $\mpcQ \in \Real^{n_z \times n_z}$, $\mpcQ \succ 0$, $\mpcR \in \Real^{m \times m}$, $\mpcR \succ 0$, $\mpcQ$ is block diagonal, and $\mpcR$ is diagonal. 
\new{The FMPC at timestep $k$ receives the estimated flat state $\hat{\flatState}_k$, has a horizon of $\mpcHorizon \in \mathbb{N}$ and uses the first flat state $\flatStateOpt$ and input $\flatInputOpt$ of its horizon as inputs to the safety filter (Section \ref{sec:safetyFilter}).}
The OCP ensures $\flatStateOpt \in \mathcal{Z}$ along the prediction horizon. 
We note that, under the same assumptions as in \cite[Sec. IV.A]{hallDifferentiallyFlatLearningBased2023}, Theorem 6.20 from \cite{gruneNonlinearModelPredictive2017} implies that for a sufficiently large $\mpcHorizon$, the FMPC renders \eqref{eq:flatSysDiscretized} asymptotically stable with respect to the reference $\flatState^{ref}$.
Furthermore, from \cite[Sec. 8.3]{borrelliPredictiveControlLinear2017}, there is a closed form solution with equivalent gain matrix $\lqrGain \in \Real^{m \times n}$
    \begin{equation}
        \flatInputOpt = - \lqrGain (\flatDisState{k} - \flatDisStateRef{k}) + \flatDisInputRef{k},
        \label{eq:equivalentGainMPC}
    \end{equation}
in the absence of constraints $\mathcal{Z}$.
Here, $\lqrGain$ can be computed from the finite horizon discrete-time Ricatti equation (DTRE).
\begin{rem}\label{rem:MPC_Ricatti_equivalence}
Note that due to the block diagonal structure we impose on $\mpcQ$ and $\mpcR$, we have that $\stabPmat$ from the solution of the DTRE will also be block diagonal.
\end{rem}

\secVspace
\subsection{System Extension}
\secVspace
\label{sec:systemExtension}
To handle flat input transformations requiring input derivatives---extending \cite{hallDifferentiallyFlatLearningBased2023}---we consider a discretized \textit{dynamic extension} \cite{isidoriNonlinearControlSystems1995},
\begin{equation}
    \dynExtState_k = \dynExtAmatDis \dynExtState_{k-1} + \dynExtBmatDis \regExtInput_k,  \quad \regDisInput = \dynExtSelectionMat \dynExtState_k.
    \label{eq:discreteExtensionDynamics}
\end{equation}
Here, $\dynExtState_k$ includes the components of the input $\regInput_k$ and their derivatives up to the required order, $\regExtInput_k$ is the extended input, and $\dynExtAmatDis$ and $\dynExtBmatDis$ fulfill the exact-discretization chain integrator dynamics.
Note that $\regExtInput_k$ may include both extended and non-extended input components. 
With this system extension \cref{eqn:flatInputPsi} simplifies to
\begin{equation}
    \flatInput_k = \vec{\Psi}(\flatState_k, \regExtInput_k) = \vec{\alpha}(\flatState_k) + \vec{\beta}(\flatState_k) \regExtInput_k.
    \label{eq:PsiLinInU}
\end{equation}
for control-affine systems such as \eqref{eq:dynamics}. 
This transformation $\vec{\Psi}$ is unknown and must be learned from data.

\begin{asm}
    We assume that the extended input $\regExtInput$, the corresponding system extension \cref{eq:discreteExtensionDynamics}, and the dimension of the flat state $\flatState$ for system \cref{eq:dynamics} are known.
\end{asm}
\new{\noindent We note that verifying the structural assumptions of our approach (flatness, relative degree, control-affine structure) requires only qualitative knowledge of the system class (e.g., from first principles or literature) and does not require detailed system identification.}
\ah{we can probably remove this second remark.}
\begin{rem}
    If \eqref{eqn:flatInputPsi} doesn't require input derivatives (i.e., the system has full relative degree) then $\regExtInput = \regInput$ and \cref{eq:PsiLinInU} holds without system extension. 
\end{rem}


%
\secVspace
\subsection{Gaussian Process Learning}
\secVspace
\label{sec:GPaffine}
We propose to learn each component of the nonlinear transformation \cref{eq:PsiLinInU} with an independent GP, as commonly done in GPMPC~\cite{hewingCautiousModelPredictive2020}. Thus, each of the $m$ GPs learns: 
\begin{equation}
    \gpFlatInput{i} = \gpFunToLearn{i}(\flatState, \regExtInput) = \alpha_i(\flatState) + \vec{\beta}_i(\flatState)^\tran \regExtInput,
    \label{eqn:gpFunToLearn}
\end{equation}
where the index $i$ on the bottom left denotes component $i$ of a vector and $\alpha_i$ and $\vec{\beta}_i$ are the components of $\vec{\alpha}$ and $\mat{\beta}$ in \cref{eq:PsiLinInU} corresponding to component $\gpFunToLearn{i}$.
We encode the affine structure of \cref{eqn:gpFunToLearn} directly in the kernel of the GP using an affine kernel, as introduced in \cite{castanedaGaussianProcessbasedMinnorm2021a},
\begin{equation}
\begin{aligned}
    \kernel(\gpInputData_i, \gpInputData_j) &= \kernAlpha(\flatState_i, \flatState_j) \\ + 
    &\regExtInput_i^\tran \, \diag (\kernBeta{1}(\flatState_i, \flatState_j), ...,  \kernBeta{m}(\flatState_i, \flatState_j) )  \, \regExtInput_j
    \label{eq:gpKernel}
\end{aligned}
\end{equation}
where $\gpInputData_i = [\flatState_i,  \regExtInput_i]$ is the input to the GPs. 
\begin{asm}
  $\kernAlpha$ and $\kernBeta{1}, ..., \kernBeta{m}$ are positive definite and bounded kernels.
  \label{asm:kernels}
\end{asm}
Given \cref{asm:kernels}, the kernel \cref{eq:gpKernel} is also positive definite and bounded as shown in \cite[Lemma 3]{castanedaGaussianProcessbasedMinnorm2021a}.
Using the affine kernel \cref{eq:gpKernel} the mean and covariance prediction at a query point $\gpInputData = [\flatState, \regExtInput]$ are computed with
\begin{align}
    \label{eq:gpMeanDefGammas}
    \gpMean{i} &= \gamOne{i}(\flatState) + \gamTwo{i}(\flatState)^\tran \regExtInput \\
    \label{eq:gpCovDefGammas}
    \gpCov{i} &= \gamThree{i}(\flatState) + \gamFour{i}(\flatState)^\tran \regExtInput + \regExtInput^\tran \gamFive{i}(\flatState) \regExtInput.
\end{align}
The derivation of each $\prescript{i}{}{\gamma_j}$ is similar to \cite{hallDifferentiallyFlatLearningBased2023}, but here we consider a multi-dimensional transformation.

\begin{asm}
    For each component $i$ of the nonlinear transformation \cref{eq:PsiLinInU}, the true function $\gpFunToLearn{i}(\gpInputData)$ belongs to a Reproducing Kernel Hilbert Space (RKHS) induced by the employed kernel \cref{eq:gpKernel} of the corresponding GP.
    Furthermore, its RKHS norm is bounded as $\parallel \gpFunToLearn{i}(\gpInputData) \parallel _{\kernel} \leq B_i < \infty$ for some $B_i > 0$.
    The observation noise of each GP is independent and uniformly bounded by $\sigma_n$.
    \label{asm:RKHSNorm}
\end{asm}
Under \cref{asm:RKHSNorm}, \cite[Thm. 6]{srinivasInformationTheoreticRegretBounds2012} guarantees that for some $\stabProb{i} \in (0, 1)$, mean predictions are bounded as
\begin{equation} \label{eq:meanBound}
        \mathrm{Pr} \{ \forall \gpInputData \in \mathcal{A}, |\gpMean{i}(\gpInputData) - \gpFunToLearn{i}(\gpInputData) | \leq \stabBetaSqrt{i} \gpStddev{i}(\gpInputData) \} \geq 1 - \stabProb{i} ,
    \end{equation} 
where $\mathcal{A}$ is compact, and $\beta$ is computed based on the properties of the GP and the training data.

\secVspace
\subsection{Safety Filter}
\secVspace
\label{sec:safetyFilter} 
We combine the learned flat-input transformation with the FMPC results in a safety filter consisting of four components: probabilistic feedback linearization, a probabilistic Lyapunov decrease constraint, a probabilistic state constraint and input constraints. We then show that this safety filter is an SOCP that finds the optimal $\regExtInput$.
To limit the search space of the optimization and to enable the stability constraint bounds, we impose $\regExtInput^{min} \leq \regExtInput \leq \regExtInput^{max}$. 
\subsubsection{Probabilistic Feedback Linearization}
From FMPC we get the desired optimal flat states and inputs $\flatStateOpt$ and $\flatInputOpt$. To find $\regDisExtInp$, we minimize the distance between the expected GP prediction and  $\flatInputOpt$
\begin{equation}
    \min_{\regDisExtInp} \; \mathbb{E} [\parallel \gpFunToLearnVec(\flatStateOpt, \regDisExtInp) - \flatInputOpt \parallel ^2 ], 
    \label{eq:fbLinOpt}
\end{equation}
where $\mathbb{E}$ is the expected value and $\gpFunToLearnVec = [\gpFunToLearn{1}, ..., \gpFunToLearn{m}]^\tran$ is the GP posterior distribution of each dimension summarized in one vector. 
%
Similar to \cite{hallDifferentiallyFlatLearningBased2023}, we use the definitions of the GP mean and covariance from \cref{eq:gpMeanDefGammas,eq:gpCovDefGammas} and leave out all terms independent of $\regDisExtInp$ to get 
\begin{equation}
\begin{aligned}
    \min_{\regDisExtInp} \, \left(\sum_{i=1}^{m} 2 (\gamOneOpt{i} - \flatInpOptComp{i}) \gamTwoOpt{i}^\tran + \gamFourOpt{i}^\tran \right) \regDisExtInp \\
    + \regDisExtInp^\tran \left(\sum_{i=1}^{m} {\gamTwoOpt{i}} \, \gamTwoOpt{i}^\tran + \gamFiveOpt{i} \right) \regDisExtInp, 
    \label{eq:probabilisticFBLin}
\end{aligned}
\end{equation}
using $\prescript{i}{}{\gamma_j}^{\hspace{-0.25em}*} \coloneqq \prescript{i}{}{\gamma_j} (\flatStateOpt) $.
The solution to this optimization finds the $\regDisExtInp$ that, when mapped through the learned nonlinear input transformation \cref{eq:PsiLinInU}, most closely matches the optimal flat input from the FMPC $\flatInputOpt$.

%
\subsubsection{Probabilistic Stability Constraint}
Using the Lyapunov decrease condition \cref{eq:ljapunovDecreasing} we formulate a probabilistic stability constraint to ensure that the input $\regDisInput$ guarantees Lyapunov decrease of the closed loop system with learned dynamics.
This will ensure that, as long as the optimization remains feasible, the tracking error will decrease.
We consider a nominal controller $\flatNomInput = - \lqrGain \trackingError{k} + \flatDisInputRef{k}$ with gain matrix $\lqrGain$ from \cref{eq:equivalentGainMPC} and a Lyapunov function $\ljapunovFun{\trackingError{k}} = \trackingError{k}^\tran \stabPmat \trackingError{k}$ with $\stabPmat$ from \cref{rem:MPC_Ricatti_equivalence}. With the nominal controller the error dynamics become $\trackingError{k+1} = \left( \dynAdis - \dynBdis \lqrGain \right) \trackingError{k} + \dynBdis \left( \gpFunToLearnVec (\flatStateOpt, \regDisExtInp) - \flatNomInput \right)$. Inserting the error dynamics into the Lyapunov decrease condition in \cref{eq:ljapunovDecreasing} results in $\stabWone^\tran \left(\gpFunToLearnVec (\flatStateOpt, \regDisExtInp) - \flatNomInput \right) + \left( \gpFunToLearnVec (\flatStateOpt, \regDisExtInp) - \flatNomInput \right)^\tran \stabWtwoMat \left( \gpFunToLearnVec (\flatStateOpt, \regDisExtInp) - \flatNomInput \right) \leq \stabWthree$
which can be written as 
\begin{align}
    &\left(\gpFunToLearnVec (\flatStateOpt, \regDisExtInp) - \flatNomInput + \stabWfour \right) ^\tran \stabWtwoMat \left( \gpFunToLearnVec (\flatStateOpt, \regDisExtInp) - \flatNomInput + \stabWfour \right) \\&\leq \stabWthree + \frac{1}{4} \stabWone^\tran \stabWtwoMat^{-1} \stabWone,\label{eq:stabCondW}
\end{align} 
where $
{\stabWone^\tran = 2 \trackingError{k}^\tran \left(\dynAdis - \dynBdis \lqrGain \right)^\tran \stabPmat \dynBdis}, 
{\stabWtwoMat = \dynBdis^\tran \stabPmat \dynBdis}, \\
{\stabWthree = \trackingError{k}^\tran \left( \stabPmat - \left(\dynAdis - \dynBdis \lqrGain \right)^\tran \stabPmat \left(\dynAdis - \dynBdis \lqrGain \right) \right) \trackingError{k} - \epsilon},\\
{\stabWfour = \frac{1}{2} \stabWtwoMat^{-1} \stabWone},
$
and $\epsilon > 0$ is a small constant to allow the inequality to be non-strict. 

\begin{thm}
    Given \cref{eq:meanBound,rem:MPC_Ricatti_equivalence} the Lyapunov decrease condition \cref{eq:stabCondW} holds when
    \begin{align}
        &\sum_{i=1}^{m} \stabWtwoMatComp{i} \left( \gpMean{i}(\flatStateOpt, \regDisExtInp) - \flatNomInpComp{i} + \stabWfourComp{i} \right)^2 + \lipschitzConst{i} \stabBetaSqrt{i} \gpStddev{i}(\flatStateOpt, \regDisExtInp)\\
        &\leq \stabWthree + \frac{1}{4} \stabWone^\tran \stabWtwoMat^{-1} \stabWone, \label{eq:stabConstBounded} 
    \end{align}    
     with probability  $\delta_s = 1 - \sum_{i=1}^{m}\delta_i$. Here $\gpMean{i}$ and $\gpCov{i}$ are the GP mean and covariance prediction of $\gpFunToLearn{i}$, 
     $\stabWtwoMatComp{i}$ is the diagonal component $i$ of $\stabWtwoMat$ 
     and $\lipschitzConst{i}$ is given by
     \begin{align}
       \lipschitzConst{i} = \max_{\mathclap{\vec{s} \in [\regDisExtInp^{min}, \regDisExtInp^{max}]}} \, 2 \stabWtwoMatComp{i} ( | \gpMean{i}(\flatStateOpt, \vec{s}) - \flatNomInpComp{i} + \stabWfourComp{i} | 
       + \bar{\rho} \gpStddev{i}(\flatStateOpt, \vec{s}) ), 
        \label{eq:lipschizLValue}
    \end{align}
    where $\bar{\rho}$ is the Gaussian quantile function at probability $\bar{\delta}$.
\end{thm}
\begin{proof}
    From \cref{rem:MPC_Ricatti_equivalence} $\stabPmat$ is a positive definite block-diagonal matrix. Combined with the block structure of $\dynBdis$ in \cref{eq:flatSysDiscretized} $\stabWtwoMat$ is a diagonal matrix $\stabWtwoMat = \diag ( \stabWtwoMatComp{1}, ..., \stabWtwoMatComp{m})$. Thus we can write the left hand side of \cref{eq:stabCondW} component wise
    \begin{equation}
        \sum_{i=1}^{m} \underbrace{\stabWtwoMatComp{i} \left( \gpFunToLearn{i} - \flatNomInpComp{i} + \stabWfourComp{i} \right)^2}_{\stabFunComp(\gpFunToLearn{i})}
        \leq \stabWthree + \frac{1}{4} \stabWone^\tran \stabWtwoMat^{-1} \stabWone,
    \end{equation}
\new{where $\stabFunComp(\gpFunToLearn{i})$ is quadratic in $\gpFunToLearn{i}$ (arguments omitted for readability). Since $\regDisExtInp$ is bounded, $\stabFunComp(\gpFunToLearn{i})$ is also bounded and $\stabFunComp$ is Lipschitz continuous with constant $\tilde{\lipschitzConst{i}} = \max_{\gpFunToLearn{i}} | 2 \stabWtwoMatComp{i} \left( \gpFunToLearn{i} - \flatNomInpComp{i} + \stabWfourComp{i} \right) | $, which we bound by $\lipschitzConst{i}$ with high probability $\bar{\delta}$. Applying \cref{eq:meanBound} gives $ |\gpFunToLearn{i} - \gpMean{i}| \leq \stabBetaSqrt{i} \gpStddev{i}$ with probability $1-\delta_i$ and Boole's inequality yields the joint probability $1 - \sum_{i=1}^m \delta_i$.}
\end{proof}

We use the definitions of the GP mean and covariance \cref{eq:gpMeanDefGammas,eq:gpCovDefGammas} to write \cref{eq:stabConstBounded} solely as a function of $\regDisExtInp$
\begingroup
\addtolength{\jot}{-0.8em}
\begin{align}    
    &\sum_{i=1}^{m} \lipschitzConst{i} \stabBetaSqrt{i} \sqrt{\gamThreeOpt{i} + \gamFourOpt{i}^\tran \regDisExtInp + \regDisExtInp^\tran \gamFiveOpt{i}\regDisExtInp} \\
    & + \regDisExtInp^\tran \left( \sum_{i=1}^{m} \stabWtwoMatComp{i} \gamTwoOpt{i} \gamTwoOpt{i}^\tran \right) \regDisExtInp  \label{eq:probStabConstraint}\\
    &\leq 
    \left(\sum_{i=1}^{m} - \stabWtwoMatComp{i}\left( 2 \gamOneOpt{i} + 2(\stabWfourComp{i} - \flatNomInpComp{i}) \right) \gamTwoOpt{i}^\tran \right) \regDisExtInp \\
    &+ \stabWthree + \frac{1}{4} \stabWone^\tran \stabWtwoMat^{-1} \stabWone 
    - \sum_{i=1}^{m} \stabWtwoMatComp{i} \left(\gamOneOpt{i} -\flatNomInpComp{i} + \stabWfourComp{i} \right)^2
\end{align}
\endgroup


%
\subsubsection{Probabilistic State Constraint}
Because the safety filter can modify the flat input, and the learned mapping $\gpFunToLearnVec (\flatStateOpt, \regDisExtInp )$ is uncertain, we need a constraint in the safety filter to ensure $\flatDisState{k+1} \in \mathcal{Z}$ with high probability. 
With the GP predictions and the discrete-time linearized dynamics \cref{eq:flatSysDiscretized} the mean of the next state becomes
$\meanState{k+1} = \dynAdis \flatStateOpt + \dynBdis [\gpMean{1}(\flatStateOpt, \regDisExtInp ), ...,\gpMean{m}(\flatStateOpt, \regDisExtInp)]^\tran$.
For brevity, we only consider uncertainty in $\meanState{k+1}$ due to $\gpCov{i}, {i \in \{1,...,m\}}$. 
Therefore the uncertainty of the next state is 
$\covState{k+1} = \dynBdis 
    \diag \left( \gpCov{1} (\flatStateOpt, \regDisExtInp), ..., \gpCov{m}(\flatStateOpt, \regDisExtInp) \right)
    \dynBdis^\tran
$.
The off-diagonal terms in the covariance matrix are zero because the components of the flat input vector are assumed to be independent of each other.
Using $\covState{k+1}$ we tighten $\mathcal{Z}$ using probabilistic reachable sets (PRS).
Following the one-step PRS tightening from \cite{hallDifferentiallyFlatLearningBased2023} and \cite{hewingCautiousModelPredictive2020}, the tightened constraint set 
            $\stateTightenedSet (\covState{k+1}) := \{ \flatDisState{k+1} | \stateSelVec^\tran \flatDisState{k+1} \leq \stateLimit - \rho(1-\delta_j) \sqrt{\stateSelVec^\tran \covState{k+1} \stateSelVec} \},$ 
for $j=1,\ldots,n_c$, ensures that under the uncertain dynamics $\meanState{k+1} \in \stateTightenedSet$ with probability level $\delta_c = 1-\sum_{j=1}^{n_c}\delta_j$.
Note $\rho(1-\delta_j)$ is the quantile function of a standard Gaussian random variable at the probability $1-\delta_j$ for some small $\delta_j$.

\label{eq:probStateConstr} 

\subsubsection{Input Constraints}
In addition to the constraints on the extended input $\regExtInput_{min} \leq \regExtInput \leq \regExtInput_{max}$ we consider constraints on the system input $\regInput$. 
This is necessary to account for actuator limits in practice.
To constrain $\regInput$ we use the system extension \cref{eq:discreteExtensionDynamics} to relate $\regExtInput$ to the bounds on $\regInput$ as 
\begin{equation}
    \regInput_{min} \leq \dynExtSelectionMat \left( \dynExtAmatDis \dynExtState_{k-1} + \dynExtBmatDis \regExtInput_k \right) \leq \regInput_{max}
    \label{eq:regInputConstraint}
\end{equation}

%
\subsubsection{Safety Filter as a SOCP}
Here we show that the safety filter takes the form of a SOCP. 
\begin{thm}
    The safety filter consisting of probabilistic feedback linearization \cref{eq:probabilisticFBLin}, a probabilistic Lyapunov decrease constraint \cref{eq:probStabConstraint}, probabilistic state constraint (\cref{eq:probStateConstr}) and input constraints \cref{eq:regInputConstraint} is written as a SOCP with 
    \begin{equation}
    \begin{aligned}
    \min_{\socpOptVar} & 
    \left[\left(\sum_{i=1}^{m} 2 (\gamOneOpt{i} - \flatInpOptComp{i}) \gamTwoOpt{i}^\tran + \gamFourOpt{i}^\tran \right), 1, 0 \right] \socpOptVar \\
    \textrm{s.t.} & \parallel \socA{i} \socpOptVar - \socB{i} \parallel \leq \socC{i} \socpOptVar + \socD{i} \quad i \in \{1, 2, 3, 4,\ldots,n_c\},\\
    &\regInput_{min} \leq \dynExtSelectionMat \left( \dynExtAmatDis \dynExtState_{k-1} + \dynExtBmatDis \regExtInput_k \right) \leq \regInput_{max},\\
    &\regExtInput^{min} \leq \regExtInput \leq \regExtInput^{max},
    \end{aligned}  
    \label{eq:socp_def}
    \end{equation}
    where $\socpOptVar = [\regDisExtInp^\tran, q_1, q_2]^\tran$, and $q_1$ and $q_2$ are auxiliary variables.
    The optimal input $\regDisExtInp$ is within the input bounds and guarantees Lyapunov function decrease and $\flatDisState{k} \in \mathcal{Z}$ with probability of at least $\delta_s + \delta_c - 1$.
\end{thm}
\ah{add lower input bounds in code if needed.}
\begin{proof}
First, we introduce $q_1$ to be able to write the cost function as a linear function in $\socpOptVar$. Take $q_1 \geq \regDisExtInp^\tran \left(\sum_{i=1}^{m} {\gamTwoOpt{i}} \, \gamTwoOpt{i}^\tran + \gamFiveOpt{i} \right) \regDisExtInp$. Using the same steps as in our previous work \cite{hallDifferentiallyFlatLearningBased2023} we can reformulate this as
    $(1 + q_1)^2 \geq 4 \regDisExtInp^\tran \left(\sum_{i=1}^{m} {\gamTwoOpt{i}} \, \gamTwoOpt{i}^\tran + \gamFiveOpt{i} \right) \regDisExtInp + (1-q_1)^2$ which can be written as a standard SOC constraint using the Cholesky decomposition of $ \gamFiveOpt{i} = \choleskyGamFive{i} \choleskyGamFive{i}^\tran$ to give expressions for $\socA{1}$, $\socB{1}$, $\socC{1}$ and $\socD{1}$.
Next, we convert the stability constraint \cref{eq:probStabConstraint} as two SOC constraints using $q_2$. We write $\regDisExtInp^\tran \left( \sum_{i=1}^{m} \stabWtwoMatComp{i} \gamTwoOpt{i} \gamTwoOpt{i}^\tran \right) \regDisExtInp \leq q_2$ as a SOC constraint with $\socB{2} = [0, ..., 0, 1]^\tran$, $\socC{2} = [\vec{0}^\tran, 0, 1]$, $\socD{2}=1$, and $\socA{2}$ is a matrix of zeros with the first $m$ columns as $\text{col}\left(\sqrt{\stabWtwoMatComp{1}} \gamTwoOpt{1}^\tran, \ldots, \sqrt{\stabWtwoMatComp{m}} \gamTwoOpt{m}^\tran, \vec{0}^T \right)$, and whose last entry is \num{-1}.
Using the subadditivity of the square root, the SOC constraint for \cref{eq:probStabConstraint} with $\socA{3}$ as a zero matrix with the first few columns as $\text{col}\left( \lipschitzConst{1} \stabBetaSqrt{1} \left[\choleskyGamFive{1}, \mat{0} \right]^T,\ldots, \lipschitzConst{m} \stabBetaSqrt{m} \left[ \choleskyGamFive{m}, \mat{0} \right]^T, \mat{0} \right)$ and the last entry as \num{1}; with $\socB{3} = \text{col}(\socBComp{3}{1},\ldots,\socBComp{3}{m},0)$ where
\begin{align}
    &\socBComp{3}{i} = \lipschitzConst{i} \stabBetaSqrt{i} \begin{bmatrix}
            - \frac{1}{2} \choleskyGamFive{i}\strut^{-1} \gamFourOpt{i}\\
            \sqrt{\frac{\gamThreeOpt{i}}{m} - \prescript{}{1}{(\frac{1}{2} \choleskyGamFive{i}\strut^{-1} \gamFourOpt{i})^2}} \\
            \smallvdots\\
            \sqrt{\frac{\gamThreeOpt{i}}{m} - \prescript{}{m}{(\frac{1}{2} \choleskyGamFive{i}\strut^{-1} \gamFourOpt{i})^2}}
        \end{bmatrix}
        = \lipschitzConst{i} \stabBetaSqrt{i} \tilde{\vec{b}_i}
\end{align}
\begin{align}
    \socC{3} &= \begin{bmatrix}
        \left(\sum_{i=1}^{m} - \stabWtwoMatComp{i}\left( 2 \gamOneOpt{i} + 2(\stabWfourComp{i} - \flatNomInpComp{i}) \right) \gamTwoOpt{i}^\tran \right) & 0 & 0
    \end{bmatrix} \\
 \socD{3} &= \stabWthree + \frac{1}{4} \stabWone^\tran \stabWtwoMat^{-1} \stabWone 
    - \sum_{i=1}^{m} \stabWtwoMatComp{i} \left(\gamOneOpt{i} -\flatNomInpComp{i} + \stabWfourComp{i} \right)^2
   \end{align}
where $\prescript{}{i}{(\cdot)^2}$ is the square of component $i$ of $(\cdot)$. We bound the terms in $\socBComp{3}{i}$ such that they remain real. 

%
For the probabilistic state constraint in \cref{eq:probStateConstr} we split up $\stateSelVec$ into components of the same size as the block vectors of $\dynBdis$ such that the constraint becomes $\stateSelVec^\tran \meanState{k+1} \leq \stateLimit - \sqrt{\sum_{i=1}^{m} \stateW{i}^2 \gpCov{i}(\flatStateOpt, \regDisExtInp)}$ where $\stateW{i} = \rho(\delta) \sqrt{\stateSelVecPart{i}^\tran \dynBdisComp{i} \dynBdisComp{i}^\tran  \stateSelVecPart{i}}$. This has the same form as the probabilistic stability constraint without auxiliary variables and can therefore be written as a SOC constraint using the same procedure as in \cite{hallDifferentiallyFlatLearningBased2023}, to give expressions for the SOC constraints $i=4,\ldots,n_c$. 
Applying Boole’s inequality yields a joint satisfaction probability of at least $\delta_s + \delta_c -1$ for both the Lyapunov decrease and flat-state constraints when applying $\regDisExtInp$ to \eqref{eq:dynamics}.
\end{proof}
\begin{rem}
Because the Lyapunov function decreases with high probability between successive sampling instants, if the problem remains feasible then $\trackingError{k} \rightarrow 0$ as $k \rightarrow \infty$. 
Although our formulation doesn't theoretically guarantee recursive feasibility, we demonstrate in \cref{sec:experiments} that, practically, feasibility is often maintained.
\end{rem}

%
\secVspace
\section{Simulation and Hardware Experiments}
\secVspace
\label{sec:experiments}
We evaluate our control approach in both simulation and hardware on an unconstrained and state/input-constrained lemniscate trajectory tracking task. To highlight our performance relative to idealized dynamics, we compare our approach with perfect-knowledge FMPC and NMPC; we also compare to GPMPC \cite{hewingCautiousModelPredictive2020, yuanSafeControlGymUnifiedBenchmark2022}  as a marker of state-of-the-art learning-based control. All code is available at \href{https://github.com/utiasDSL/mimo_fmpc_socp}{github.com/utiasDSL/mimo\textunderscore fmpc\textunderscore socp}. 

The 2D quadrotor model has state $\regState = [x, \dot{x}, z, \dot{z}, \theta, \dot{\theta}]^\tran$ and input $\regInput = [T_c, \theta_c]^\tran$. Here $(x, y)$ is the position, $\theta$ is the roll angle, $T_c$ is the thrust, and $\theta_c$ is the commanded roll angle. The system dynamics are $\ddot{x} =  \sin{\theta}(\beta_2 + \beta_1 \, T_c)$, $\ddot{z} = \cos{\theta}(\beta_2  + \beta_1 \, T_c) - g$ and $\ddot{\theta} = \alpha_1 \theta + \alpha_2 \dot{\theta} + \alpha_3 \theta_c$.
We used a Bitcraze Crazyflie 2.1 as our hardware platform; its identified model parameters can be found in our code.
This system is flat with flat output $\flatOutput = (x, z)^\tran$, and the flat state vector is $\flatState = [x, \dot{x}, \ddot{x}, \dotThree{x}, z, \dot{z}, \ddot{z}, \dotThree{z}]^\tran$. As the system has non-full relative degree, we extend the system with a double integrator on $T_c$ and the extended input becomes $\regExtInput = [ \ddot{T_c}, \theta_c]^\tran$.

\vspace{-1mm}
\subsection{Simulation Experiment and Results}
All controllers were tuned to maximize performance while maintaining reasonable computational efficiency.
All controllers are run at \SI{100}{\hertz} with a horizon of \num{50} using \texttt{ACADOS} \cite{acados}, \new{ a state-of-the-art solver for embedded optimal control, ensuring a fair comparison across formulations}. 
The gains $\mpcQ$ and $\mpcR$ for each method were kept as similar as possible while optimizing individual performance. 
All GPs used squared exponential kernels and were trained offline on 600 data points.
For GPMPC, two GPs were used to model error residuals to account for incorrect prior parameters. 
The first GP corrects $\beta_2 + \beta_1 T_c$ and the second corrects $\ddot{\theta}$.
To keep GPMPC tractable, sparse GP inference with \num{75} inducing points was used.
For our FMPC+SOCP method, the points $\{\flatState_k, \regExtInput_k, \flatInput_k \}$ were sampled around the reference.
In both the constrained and unconstrained tasks, the simulation was run \num{30} times starting at a randomized state within a small ball around the reference. \new{On held-out data, all GPs achieved relative RMSE below 2\% of the output range, with over 96\% of test points falling within the 2-$\sigma$ confidence bounds, except for the first flat input component (86\%), which is stiff and more difficult to calibrate.}


\begin{figure*}[t]
  \vspace{-2mm}
  \centering
	\subfloat[Simulation]{\centering
	    \centering
	    \includegraphics[width = 0.22 \linewidth]{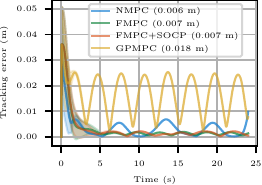}
	    \label{fig:tracking_error_fast}
    }
    \subfloat[Hardware]{\centering
        \centering
	    \includegraphics[width = 0.22 \linewidth]{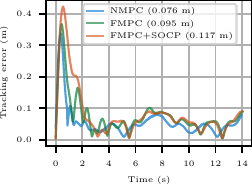}
        \label{fig:tracking_error_real}
     }
     \subfloat[Simulation]{\centering
	    \includegraphics[width = 0.22 \linewidth]{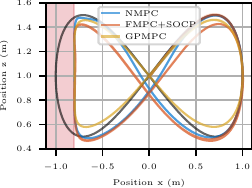}
	    \label{fig:traj_constrained}}
    \subfloat[Hardware]{\centering
	    \includegraphics[width = 0.22 \linewidth]{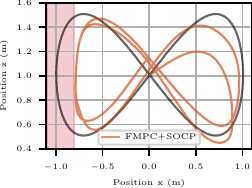}
	    \label{fig:traj_constrained_hardware}}
     
   \caption{Comparisons of the position tracking errors on the unconstrained flat state tracking task for the (a) simulated and (b) hardware show that our formulation is very competitive with the controllers using the true dynamics.
   When including constraints, we can see the performance is still competitive in (c) simulation, and that on (d) hardware our controller still successfully respects the constraint.
   All simulation plots show the mean (solid) and 95\% confidence interval (shaded) across all the initial conditions. The red shaded region represents the flat state constraint.
   \new{The root-mean-square error of each controller in the unconstrained case is given in the legend inside the parenthesis.}
   }
  \label{fig:fast}
  \vspace{-5mm}
\end{figure*}

\begin{figure}[t]
	 \centering
        \includegraphics[trim={0 2mm 0 2mm},clip]{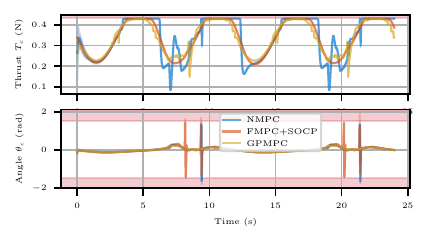}
	    
   \vspace{-3mm}
  \caption{The mean input (solid) and 95\% confidence interval (shaded) across all the initial conditions for the constrained simulation task. The red shaded region represents the constraint.}
  \label{fig:inputs_constrained}
\end{figure}

In \cref{fig:tracking_error_fast}, we show a comparison of the tracking error without constraints. After a transient phase, our approach with purely learned dynamics achieves tracking performance similar to FMPC and slightly better than NMPC and GPMPC.  
The observed transient error arises from the limits on $\regExtInput$, which restricts rapid thrust changes. These are even more significant when starting from hover, as seen in the hardware comparison in \cref{fig:tracking_error_real}.

In \cref{fig:traj_constrained}, we demonstrate the behavior under constraints on the flat state $x$ and on the input $T_c$. Our approach anticipates the state constraint on the FMPC horizon and thus satisfies it with minimal tracking error. 
\new{In contrast, the input constraint on $T_c$ is enforced in the single-step safety filter (\cref{fig:inputs_constrained}), so it cannot be anticipated. This can cause a slight overshoot bounded by $\Delta T_c = ({\dot{T}_{c, 0}}^2)/({2 \ddot{T}_{c, max}})$, where $\dot{T}_{c, 0}$ is the thrust derivative when the constraint becomes active and $\ddot{T}_{c, max}$ is the bound on the corresponding component of $\regDisExtInp$; we tighten the constraint by this amount. Minor residual violations can also arise from \texttt{ACADOS} solver instabilities, as observed for both NMPC and our method.}

\begin{figure}[t]
\vspace{-3mm}
	 \centering
        \includegraphics{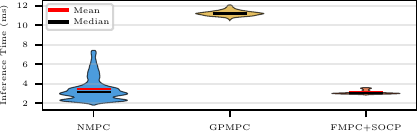}
	    
   \vspace{-3mm}
  \caption{Violin plot of the average inference times on the constrained task. Outliers were filtered using \num{1.5} times the interquartile range.
  }
  \label{fig:timings}
\end{figure}

The distribution of inference times on the constrained task of each approach are shown in \cref{fig:timings}.
We see that our method averages \SI{3.1}{\milli \second}, which is less than \SI{3.4}{\milli \second} for NMPC and more than three and a half times faster than \SI{11.2}{\milli \second} for GPMPC. 
The larger variance in NMPC and GPMPC comes from the nonlinear optimizations that sometimes---particularly near constraint boundaries, far away from warm-start solutions, or in uncertain regions---require significant compute.
These results confirm that our learning-based flat control delivers significant computational benefits without sacrificing performance, even when accounting for dynamics uncertainty and constraints.
\subsection{Hardware Experiments and Results} 
Hardware experiments were performed using a Bitcraze Crazyflie 2.1 quadrotor.
A Vicon motion capture system with an extended Kalman filter was used to estimate the system state $\regState$ and flat state $\flatState$ assuming perfect system knowledge.
The GPs trained in simulation were used for the hardware experiments---training the GPs on real-world data is left for future work. 
The controllers were run at \SI{50}{\hertz} on a standard laptop from a base station and telemetry was wirelessly sent to the drone.
Instead of starting near the reference as in the simulation tasks, the drone started from hover.

For the unconstrained experiment, the tracking error in \cref{fig:tracking_error_real} indicates that our approach is competitive with the NMPC and FMPC.
On the constrained tracking task, our proposed approach respects all constraint boundaries and remains feasible for all time (\cref{fig:traj_constrained_hardware}). Due to model mismatch, state-estimation noise, and the use of GPs trained solely in simulation, the tracking errors are larger than in simulation for all approaches. 
\new{Offset-free MPC techniques~\cite{MORARI20122059} could further reduce steady-state bias in hardware and are left for future work.}



%
\secVspace
\section{Conclusion}
\secVspace
%
We introduced a computationally efficient learning-based control method for nonlinear multi-input, control-affine differentially flat systems that guarantees probabilistic Lyapunov decrease and probabilistic half-space flat state constraint satisfaction and respects input constraints.
While recursive feasibility is not yet theoretically guaranteed, we show through simulation and hardware experimentation that our controller remains practically feasible, stable, and respects constraints. 
\new{Scalability is aided by the convexity and block-diagonal decoupling of both subproblems, and per-dimension GP inference.}
Overall, our experimental results show that our flatness-exploiting learning-based controller is significantly more efficient than GPMPC, the gold-standard learning-based controller, and doesn't compromise on performance.
This work lays a strong foundation for future work on learning-based FMPC with theoretical recursive feasibility and asymptotic stability under constraints.




%

\bibliographystyle{IEEEtran}
\bibliography{IEEEabrv,Paper-MA}
\end{document}